\documentclass[journal,letterpaper,twocolumn,twoside,nofonttune]{IEEEtran}

\usepackage{mathtools}
\usepackage{amsmath}
\usepackage{amssymb}
\usepackage{amsmath}
\usepackage{cite}
\usepackage{algorithm}
\usepackage{enumerate,multirow}
\usepackage{algpseudocode}

\usepackage{siunitx}
\usepackage{multirow}
\usepackage{multicol}
\usepackage{color}
\usepackage{xcolor}

\usepackage{comment}

\usepackage{geometry}
\usepackage{graphicx}
\usepackage{graphicx}
\date{}

\usepackage{calc}
\newcolumntype{M}[1]{>{\centering\arraybackslash}m{#1}}
\newcolumntype{N}{@{}m{0pt}@{}}

\usepackage{mdwmath}
\usepackage{blindtext}
\usepackage{eqparbox}
\usepackage{fixltx2e}
\usepackage{stfloats}

\IEEEoverridecommandlockouts
\title{\Huge
{High-Rate Fair-Density Parity-Check Codes}}
\author{\large%
Hessam Mahdavifar
\thanks{H.\ Mahdavifar is with the Department of Electrical and Computer Engineering, Northeastern University, Boston, MA 02115  (email: h.mahdavifar@northeastern.edu.)}
}

\newtheorem{theorem}{{Theorem}}
\newtheorem{lemma}[theorem]{{Lemma}}
\newtheorem{proposition}[theorem]{{Proposition}}

\newtheorem{definition}{{Definition}}

\newcommand{\cA}{{\cal A}} 

\newcommand{\cC}{{\cal C}} 

\newcommand{\cE}{{\cal E}} 

\newcommand{\cG}{{\cal G}}

\newcommand{\cS}{{\cal S}}

\newcommand{\cV}{{\cal V}}

\DeclareMathAlphabet{\mathbfsl}{OT1}{ppl}{b}{it} 

\newcommand{\bc}{\mathbfsl{c}}

\newcommand{\by}{\mathbfsl{y}}

\newcommand{\be}[1]{\begin{equation}\label{#1}}
\newcommand{\ee}{\end{equation}} 
\newcommand{\eq}[1]{(\ref{#1})}

\renewcommand{\leq}{\leqslant}
 
\renewcommand{\geq}{\geqslant}

\newcommand{\Pref}[1]{Pro\-po\-si\-tion\,\ref{#1}}
\newcommand{\Lref}[1]{Lem\-ma\,\ref{#1}}
\newcommand{\Cref}[1]{Co\-ro\-lla\-ry\,\ref{#1}}

\newcommand{\deff}{\mbox{$\stackrel{\rm def}{=}$}}

\begin{document}

\vspace{10mm}
\maketitle

\begin{abstract}

We introduce fair-density parity-check (FDPC) codes targeting \textit{high-rate} applications. In particular, we start with a base parity-check matrix $H_b$ of dimension $2 \sqrt{n} \times n$, where $n$ is the code block length, and the number of ones in each row and column of $H_b$ is equal to $\sqrt{n}$ and $2$, respectively. We propose a deterministic combinatorial method for picking the base matrix $H_b$, assuming $n=4t^2$ for some integer $t \geq 2$. We then extend this by obtaining permuted versions of $H_b$ (e.g., via random permutations of its columns) and stacking them on top of each other leading to codes of dimension $k \geq n-2s\sqrt{n}+s$, for some $s \geq 2$, referred to as order-$s$ FDPC codes. We propose methods to explicitly characterize and bound the weight distribution of the new codes and utilize them to derive union-type approximate upper bounds on their error probability under Maximum Likelihood (ML) decoding. For the binary erasure channel (BEC), we demonstrate that the approximate ML bound of FDPC codes closely follows the random coding upper bound (RCU) for a wide range of channel parameters. Also, remarkably, FDPC codes, under the low-complexity min-sum decoder, improve upon 5G-LDPC codes for transmission over the binary-input additive white Gaussian noise (B-AWGN) channel by almost 0.5dB (for $n=1024$, and rate\ $=0.878$). Furthermore, we propose a new decoder as a combination of weighted min-sum message-passing (MP) decoding algorithm together with a new progressive list (PL) decoding component, referred to as the MP-PL decoder, to further boost the performance of FDPC codes. 

This paper opens new avenues for a fresh investigation of new code constructions and decoding algorithms in high-rate regimes suitable for ultra-high throughput (high-frequency/optical) applications. 

\end{abstract}

\section{Introduction} 
\label{sec:Introduction}

Codes defined over sparse structures, e.g., low-density parity-check (LDPC) codes, demonstrate excellent near-capacity performances \cite{mackay1999good,richardson2001capacity} and have had a huge impact on the development of wireless communications systems. They currently serve as the channel coding mechanism for the main data channel in the fifth-generation (5G) systems \cite{richardson2018design}. In general, LDPC codes, as defined in the original work by Gallager \cite{gallager1962low}, are specified via a certain sparse structure on the parity-check matrix, e.g., by setting the number of ones in each row and column to be a small fixed number. Consequently, LDPC codes are often designed for fixed \textit{moderate-rate} codes, e.g., the code rate of 5G-LDPC codes is approximately between $1/3$ and $8/9$ \cite{hamidi2018analysis}. 

Besides LDPC codes, majority of other state-of-the-art codes are also designed with decoders that are often optimized at moderate rates in order to address diverse and extreme scenarios in wireless systems. However, ultra-high data rate mmWave and THz links, envisioned for 5G systems and beyond, with precise beam alignment and interference cancellation at the front end lead to an effective high-SNR regime from the decoder’s perspective \cite{mo2014high,song2022robust}. Also, optical channels exhibit very high SNRs with already low effective bit-error rate (BER) yet necessitating forward error correction (FEC) in order to ensure the target BER requirements are met (typically less than $10^{-9}$ and can be as low as $10^{-15}$) \cite{graell2020forward}. Such high-throughput communication scenarios also require very low-latency and low-overhead coding mechanisms rendering various coding schemes with advanced and complex decoders inapplicable.

\begin{figure}[t]
	\centering
	\includegraphics[width=\linewidth]{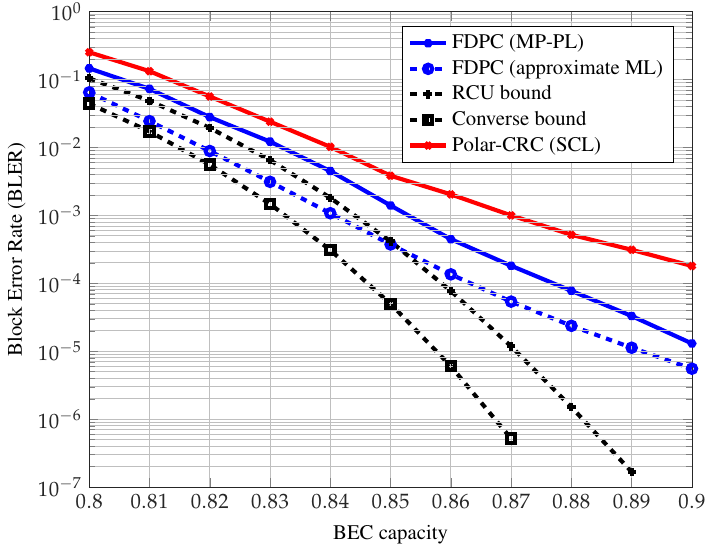}
	\caption{Performance comparison over BECs ($n=256, k=195$): FDPC codes versus fundamental bounds as well as polar-CRC with SCL decoder.}
	\label{BEC-plot}
	\vspace{-0.2in}
\end{figure}

In this paper, we introduce fair-density parity-check (FDPC) codes targeting high-rate applications. In Figure\,\ref{BEC-plot}, we demonstrate the performance of a $(256,195)$ FDPC code over BECs with varying erasure probability between $0.1$ and $0.2$. For comparison, we use polar codes \cite{Arikan} with state-of-the-art successive-cancellation list (SCL) decoder aided with a cyclic redundancy check (CRC) decoder \cite{TV} (CRC$=8$, list size $L=32$). Polar-CRC codes are picked as the coding scheme for 5G control channels, due to their superior performance with SCL decoder in the short-to-medium block length regimes \cite{bioglio2020design}. We observe that FDPC codes perform better than polar codes in the entire range of considered channel parameters, e.g., by more than one order of magnitude in the block error rate for BEC$(0.1)$. Furthermore, an approximate ML bound for the FDPC codes, averaged over their ensemble, is derived via analytical characterization of their weight distribution. In Figure\,\ref{BEC-plot}, we also show the random coding upper bound (RCU bound) and the converse bound by Polyanskiy \textit{et al.} \cite{poly}. The FDPC codes closely follow the RCU bound up to the capacity $0.85$ but start diverging after that due to having a small number of low-weight codewords.

Note that a simplified approximation, also referred to as the \textit{dispersion bound}, is often used to measure the gap between the performance of practical channel codes and what is fundamentally achievable, e.g., with random codes. The dispersion bound provides a closed form approximation for both the RCU and the converse bound \cite{poly}, presuming that these two bounds are very close. However, when the code rate is close to zero or one, the approximation may become inaccurate. This phenomena is studied in depth for low-rate channel coding in \cite{fereydounian2023channel}, and, as we can see in Figure\,\ref{BEC-plot}, the RCU and converse bound can actually differ by one order of magnitude especially as the capacity gets closer to one.

\begin{figure}[t]
	\centering
	\includegraphics[width=\linewidth]{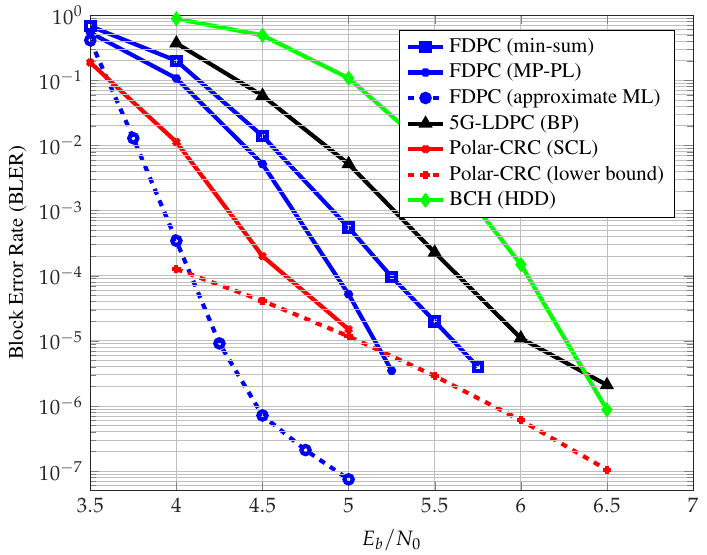}
	\caption{Performance comparison over B-AWGN ($n=1024$, rate\ $=0.878$): FDPC codes versus 5G-LDPC, polar-CRC and BCH codes.}
	\label{AWGN-plot}
	\vspace{-0.2in}
\end{figure}

In Figure\,\ref{AWGN-plot}, we show the performance of a $(1023,898)$ FDPC code, constructed with $n=1024$ and shortened by $1$ bit due to weight distribution considerations to be discussed in Section\,\ref{sec21}, over binary-input additive white Gaussian noise (B-AWGN) channels. The FDPC code with the \textit{plain} min-sum decoder (i.e., no offsetting, weighting, nor parameter tuning) performs better than the 5G-LDPC code under the belief propagation (BP) decoding by about $0.5$\,dB. Furthermore, the LDPC code starts to show an error floor at block error rate (BLER) $\approx 10^{-5}$, but the FDPC code is not expected to hit an error floor till $10^{-6}$ BLER or lower (given the simulations and the ML approximation curve). In fact, we have the exact characterization of the minimum distance and the weight distribution for the low-weight codewords of FDPC codes allowing us to precisely predict the error floor performance. 

It is also observed that the performance of the FDPC code can be boosted since there is a considerable gap between the min-sum performance curve and the approximate ML curve. To this end, we propose a parallelizable and low-complexity variation of min-sum message passing (MP) decoder coupled with a newly proposed progressives list (PL) decoder component, referred to as the MP-PL decoder. The FDPC code with MP-PL decoder can even beat polar-CRC with SCL decoder (CRC$=8$, $L=32$) in lower SNRs at BLER$<10^{-5}$. For this comparison, we use a lower bound on the performance of polar-CRC which is derived by only counting the number of codewords with weight equal to the minimum distance. Also, at this BLER, the MP-PL decoder is, on average, $50$ times less complex and $800$ times faster in latency, compared with SCL, the details of which will be discussed in Section\,\ref{sec:comp}. Note that the lower bound for the polar-CRC is for any decoder, including the ML decoder, and indicates that with ML decoder FDPC code is a better code than polar-CRC at least for BLER$<10^{-4}$. 

The performance of BCH code with hard-decision decoding (HDD) is also shown in Figure\,\ref{AWGN-plot}, that is about 1\,dB worse than that of FDPC codes. There are soft-decision based decoders that have been utilized for BCH codes, e.g., ordered statistics decoders \cite{fossorier1995soft}, however, such decoders are often only used in shorter block lengths and lower rates due to complexity reasons \cite{van2016performance}.

The rest of this paper is organized as follows. In Section\,\ref{sec:two} the details of the FDPC code construction are discussed. In Section\,\ref{sec:three}, we study the weight distribution of FDPC codes and provide bounds on their ML decoding performance. And, in Section\,\ref{sec:dec} we describe the proposed MP-PL decoder, characterize its complexity and latency. Further simulation results are presented in Section\,\ref{sec:sim} to conclude the paper.

\section{Code Construction}
\label{sec:two}

\subsection{Base Matrix}
\label{sec:base}

Let $n=4t^2$, for some integer $t\geq 2$. Then the base parity-check matrix $H_b$ is constructed as follows.

\begin{definition}
\label{hb-def}
The matrix $H_b$ consists of all binary column-vectors of length $4t$ and Hamming weight two with the indices of the two non-zero entries differing by an odd number. 
\end{definition}

Note that all rows of $H_b$ have weight $2t$. 

\noindent{\textbf{Example 1.}} For $t=2$, the matrix $H_b$ is given as follows:
\begin{align*}
\footnotesize{
\left[
\begin{array}{llllllllllllllll}
1 & 1 & 0 & 0 & 0 & 0 & 0 & 0 & 0 & 0 & 0 & 1 & 0 & 0 & 1 & 0 \\
1 & 0 & 1 & 1 & 0 & 0 & 0 & 0 & 0 & 0 & 0 & 0 & 0 & 1 & 0 & 0 \\
0 & 0 & 1 & 0 & 1 & 1 & 0 & 0 & 0 & 0 & 0 & 0 & 0 & 0 & 0 & 1 \\
0 & 1 & 0 & 0 & 1 & 0 & 1 & 1 & 0 & 0 & 0 & 0 & 0 & 0 & 0 & 0 \\
0 & 0 & 0 & 1 & 0 & 0 & 1 & 0 & 1 & 1 & 0 & 0 & 0 & 0 & 0 & 0 \\
0 & 0 & 0 & 0 & 0 & 1 & 0 & 0 & 1 & 0 & 1 & 1 & 0 & 0 & 0 & 0 \\
0 & 0 & 0 & 0 & 0 & 0 & 0 & 1 & 0 & 0 & 1 & 0 & 1 & 1 & 0 & 0 \\
0 & 0 & 0 & 0 & 0 & 0 & 0 & 0 & 0 & 1 & 0 & 0 & 1 & 0 & 1 & 1 
\end{array}
\right].
}
\end{align*}

\begin{lemma}
\label{lem-hb}
The binary rank of $H_b$ is $4t-1$. Furthermore, the code $\cC_b$ with $H_b$ as its parity-check matrix has minimum distance $d_{\min}(\cC_b) = 4$. 
\end{lemma}

\begin{proof}
Note that the sum of all rows of $H_b$ equals zero. Hence, rank$(H_b) \leq 4t-1$. Now, suppose that the sum of a subset $\cS$ of rows of $H_b$ is zero. Consider the tanner graph $\cG$ representing the parity-check matrix $H_b$. Then any two variable nodes in $\cG$ connecting via a check node must either both appear in $\cS$ or both not appear in $\cS$ (since the column weights of $H_b$ are exactly two). Since $\cG$ is connected, then the only non-trivial possibility for $\cS$ is the set of all rows of $H_b$. This implies rank$(H_b) = 4t-1$. 

With regards to $d_{\min}(\cC_b)$, columns of $H_b$ are distinct and no three of them sum up to zeros. The latter is due the property of $H_b$ that the non-zero entries in any of its columns differ by an odd number. There are selections of four columns of $H_b$ summing to zeros. Hence, $d_{\min}(\cC_b) = 4$. 
\end{proof}

\subsection{Expanding via random permutations and cascading}

Let $s\geq 2$. Then the FDPC code of order $s$ is defined by applying $s-1$ different permutations to the base matrix $H_b$, where permutations are applied column-wise, and stacking them on top of each other in order to obtain the overall parity-check matrix $H$. More specifically, let $\pi_1,\dots,\pi_{s-1}$ be distinct permutations on the set $[n]:=\{1,2,\dots,n\}$. Then the parity-check matrix $H$ for the order-$s$ FDPC code associated with this set of permutations is given as follows:
\begin{align}
\label{H-def}
H = 
\left[
\begin{array}{l}
\hspace{5mm}H_b\\
\hspace{2mm}\pi_1(H_b)\\
\hspace{6mm}.\\
\hspace{6mm}.\\
\hspace{6mm}.\\
\pi_{s-1}(H_b)
\end{array}
\right]_{r \times n},
\end{align}
where $r=2s\sqrt{n}=4st$. Note that the rows of $H$ may not be linearly independent. Hence, the dimension of an order-$s$ FDPC code is greater than or equal to $n - r + s$. In this paper, we focus on order-$2$ FDPC codes for the most part. In this case, we simply work with one permutation $\pi$, and the parity-check matrix $H$ that is the result of stacking $H_b$ and $\pi(H_b)$ on top of each other. 

\section{Weight Distribution and ML Bound}
\label{sec:three}

\subsection{Weight Distribution}
\label{sec21}

We first derive the weight distribution for the order-$1$ FDPC code $\cC_b$ with $H_b$ as its parity-check matrix. To this end, consider the Tanner graph $\cG_b$ representing $H_b$, where variable nodes are indexed by $1,2,\dots,n$, and check nodes are indexed by $1,2,\dots,r$, where $r=4t$. Then the following lemma describes a certain property of codewords in $\cC_b$.

\begin{lemma}
\label{lemma-loop}
Let $\bc = (c_1,c_2,\dots,c_n)$ be a non-zero codeword of weight $w$ in $\cC_b$. Let $I_{\bc} = \{i_1,i_2,\dots,i_w\}$ denote its \textit{support}, i.e., the indices of non-zero entries in $\bc$, and consider the sub-graph $\cG'_b \subseteq \cG_b$ induced by variable nodes $i_1,i_2,\dots,i_w$ and their neighbors. Then $\cG'_b$ is either a loop or a collection of non-overlapping loops. 
\end{lemma}
\begin{proof}
Note that the columns of $H_b$ indexed by indices in $I_{\bc}$ sum up to zeros. Now, consider $i_j \in I_{\bc}$ and let $j_1$ and $j_2$ denote the indices of non-zero entries in the $i_j$-th column of $H_b$. Then there must be at least one other variable node in $I_{\bc}$ (more precisely, an odd number of them), other than $i_j$, that is connected to the check node $j_1$. The same is true for the check node $j_2$. Then we follow the same argument for these two until we eventually reach a loop (since some variable node will be eventually revisited). This loop is contained in $\cG'_b$, and we can remove it from $\cG'_b$ and repeat the same argument for what is left in $\cG'_b$. This concludes the proof. 
\end{proof}

\noindent{\textbf{Example 2.}} Consider the code $\cC_b$ with parameter $t=2$, as presented in Example 1. Then the vector $(1,1,1,0,1,0,0,0,0,0,0,0,0,0,0,0)$ is a codeword in $\cC_b$. The sub-graph induced by columns indexed by $1,2,3,5$ consists of all edges connecting variable nodes $1,2,3,5$ and check nodes $1,2,3,4$, which is a loop of length $8$. 

\begin{definition}
We say that a codeword in the order-$1$ FDPC code $\cC_b$ is \textit{irreducible} if its corresponding induced sub-graph $\cG'_b$, as discussed in \Lref{lemma-loop}, consists of a single loop. For instance, the codeword presented in Example\,2 (or any other codeword of weight $4$ or $6$) is an irreducible codeword. 
\end{definition} 

Note that if a codeword $\bc$ is not irreducible, then it can be written as a sum of irreducible codewords with non-overlapping supports. We refer to this collection of irreducible codewords as a \textit{base} for the codeword $\bc$, which is unique for each codeword. Note also that there are irreducible codewords of any even weight between $4$ and $4t$, which is the maximum possible length for irreducible codewords. The next lemma presents the exact number of irreducible codewords of any given weight.

\begin{lemma}
\label{irrd-lemma}
Let $m$ with $4 \leq m \leq 4t$ be an even integer. Then there exists
\be{A_ir_def}
A^{(\text{ir})}_m \,\deff\,\prod_{i=0}^{m/2-1} (2t-i)^2/m
\ee
irreducible codewords of weight $m$ in $\cC_b$. For any other $m$, there exists no irreducible codeword of weight $m$.
\end{lemma}
\begin{proof}
The proof is by following straightforward counting arguments. 
\end{proof}

For $w\in [n]$, let $A_w$ denote the number of codewords of weight $w$ in the code $\cC_b$. The next lemma builds upon \Lref{irrd-lemma} in order to obtain bounds on the weight distribution coefficients $A_w$'s. 
\begin{proposition}
\label{weight-ub}
For any $n$ and $w \leq n$ we have
\be{weight-ub-eq}
A_w \leq \sum_{a_m\text{'s}} \prod_{m} {A^{(\text{ir})}_m \choose a_m},
\ee
where the summation is over all non-negative integers $a_4,a_6,\dots,a_{4t}$ with $\sum_{m=4}^{4t} ma_m = w$. 
\end{proposition}
\begin{proof}
The upper bound is essentially obtained by upper bounding the number of all possible collections of irreducible codewords that form a base for a codeword of weight $w$. 
\end{proof}

Note that the upper bound in \Pref{weight-ub} can be improved by excluding the counted instances of collections of irreducible codewords that some of them overlap with each other. Taking into account such an improvement results in a cumbersome, yet computable, upper bound which we use while presenting the numerical upper bounds on the performance of ML decoder utilizing the upper bounds on the weight distribution of the code. 

Next, consider the ensemble of all FDPC codes of length $n=4t^2$ and order $s$. In particular, we consider $s=2$ but the results on the weight distribution can be extended to general $s$ in a straightforward fashion. 

Suppose that the permutation $\pi$ in the construction of the order-$2$ FDPC code is selected uniformly at random, and let $\cA^{(2)}_w$ denote the random variable representing the number of codewords of weight $w$ in the resulting code. Then the expected value of the weight, denoted by $\cE\{\cA^{(2)}_w\}$, averaged over all permutations is given as follows:
\be{eq-avgA}
\cE\{\cA^{(2)}_w\} = \frac{A_w^2}{{n \choose w}}.
\ee

\noindent{\textbf{Example 3.}} For $w=4$, combining \eqref{eq-avgA} with \Lref{irrd-lemma} and straightforward computations leads to 
\be{eq-avgA4}
\cE\{\cA^{(2)}_4\} = \frac{3n(\sqrt{n}-1)^4}{2(n-1)(n-2)(n-3)} < \frac{3}{2}. 
\ee
And, similarly, for $w=6$ we have
\be{eq-avgA6}
\cE\{\cA^{(2)}_6\} = \frac{20n(\sqrt{n}-1)^4(\sqrt{n}-2)^4}{(n-1)(n-2)(n-3)(n-4)(n-5)} < 20. 
\ee

\noindent{\textbf{Remark 1.}} Example\,3 demonstrates that the average number of low-weight codewords in the random ensemble of FDPC codes is a small constant. This can be leveraged to easily increase the minimum distance of the code by shortening it, something that can significantly improve the code performance in the error floor region. For instance, by carefully shortening the code by $1$ bit, we obtain FDPC codes of minimum distance at least $6$ with high probability. Note that shortening an $(n,k)$ code by $t$ bits results in an $(n-t,k-t)$ code, which may have a negligible effect on the rate for high-rate codes and small values of $t$. For instance, shortening an order-2 FDPC code of length $2^{10}$, a $(1024,899)$ code, by $20$ bits results in a $(1004,879)$ code reducing the rate only by $0.0024$ from $0.8779$ to $0.8755$. The FDPC code with the performance shown in Figure\,\ref{AWGN-plot} is shortened by $1$ bit to eliminate its only weight-$4$ codeword, resulting in a $(1023,898)$ code. 

\noindent{\textbf{Remark 2.}} In this remark, we discuss the process of shortening the codes in the ensemble more explicitly. For instance, consider the weight $w=4$ and suppose we aim at shortening the FDPC codes by $\alpha_4$ bits, corresponding to weight-$4$ codewords. For each FDPC code $\cC$ in the ensemble with at least $\alpha_4$ codewords, we remove one column from the parity-check matrix with the index corresponding to a non-zero entry of one of the weight-$4$ codewords. This shortening process will reduce the number of weight-$4$ codewords by at least $\alpha_4$. If the number of weight-$4$ codewords in an FDPC code in the ensemble is less than $\alpha_4$, one can still shorten the code by $\alpha_4$ bits and ensure that no weight-$4$ codeword remains in the code. One can continue the shortening process, in a similar fashion, with $\alpha_6$ bits corresponding to weight-$6$ codewords, etc.  

\subsection{ML Bound}

In this section, we derive bounds on the probability of error of the maximum likelihood (ML) decoder for FDPC codes. We start by considering communication over BEC$(\epsilon)$. The following union-type upper bound holds on the probability of error of the ML decoder, denoted by $P_e^{(\text{ML})}$, for any linear code $\tilde{\cC}$ \cite[Chap. 2]{viterbi2013principles}:
\be{Pe_BEC}
P_e^{(\text{ML})} \leq \sum_{w=d_{\min}}^{n} \tilde{A}_w \epsilon^w,
\ee
where $\{\tilde{A}_w: w=d_{\min},\dots,n\}$ denote the weight distribution of $\tilde{\cC}$, and $d_{\min}$ is the minimum distance of $\tilde{\cC}$. The reasoning behind this bound is as follows. Given an erasure pattern $\cE$ the ML decoder fails if and only if there is a non-zero codeword $\bc$ in $\tilde{\cC}$ whose support is a subset of the erasure positions. And the bound in \eq{Pe_BEC} is a union bound on the probability that the set of erasure positions contain a codeword, where the union bound is taken with respect to all possible codewords. However, this upper bound, with all the terms therein, may become numerically inaccurate to compute as each term $\tilde{A}_w \epsilon^w$ is the product of two terms, with one decaying exponentially small and then other one growing exponentially large. Therefore, instead of the whole expression in \eq{Pe_BEC}, a common approach is to consider a threshold $w_t$ and compute the expression only up to the $w_t$-th term. For a properly chosen value of $w_t$, this gives a fairly precise approximation for the performance of the ML decoder. This is often referred to as an approximate union bound on the performance of ML decoder.

For the ensemble of FDPC codes, say the order-2 one, one could take the expected value from both sides of \eq{Pe_BEC} and arrive at
\be{Pe_BEC2}
\cE\{P_e^{(\text{ML})}\} \leq \sum_{w=4}^{n} \cE\{\cA^{(2)}_w\} \epsilon^w.
\ee

The next lemma discusses how this upper bound can be revised for the ensemble of shortened FDPC codes. 

\begin{proposition}
\label{proposition_ml}
For $w=4,6,\dots,2(d-1)$, for some integer $d>2$, consider integers $\alpha_w \geq \lfloor \cE\{\cA^{(2)}_w\} \rfloor$. Consider the ensemble of shortened FDPC codes by $\sum_{i=2}^{d-1} \alpha_{2i}$ bits and minimum distance $d$, where $\alpha_{2i}$ columns of $H$ are eliminated reducing the number of weight-$2i$ codewords by $\alpha_{2i}$. Let 
\be{gamma-def}
\gamma = \sum_{i=2}^{d-1} \frac{\cE\{\cA^{(2)}_{2i}\}}{\alpha_{2i}+1}.
\ee
and assume that $\gamma < 1$. Then the average probability of ML decoder over the ensemble of FDPC codes with $d_{\min} \geq 2d$ is upper bounded as follows:
\be{Pe_BEC3}
\cE\{P_e^{(\text{ML})}\} \leq \frac{1}{1-\gamma} \sum_{w=2d}^{n} \cE\{\cA^{(2)}_w\} \epsilon^w.
\ee
\end{proposition}
\begin{proof}
For any $w$, using Markov inequality, the probability that a random FDPC code has more than $\alpha_w$ codewords of weight $w$ is upper bounded by $\frac{\cE\{\cA^{(2)}_{w}\}}{\alpha_{w}+1}$. Then by taking the union bound across all $w$'s, for $w=4,6,\dots,2(d-1)$, the probability that an FDPC code has a codeword of weight less than $2d$, after shortening by $\sum_{i=2}^{d-1} \alpha_{2i}$ bits, is upper bounded by $\gamma$. Hence, at least a fraction of $1-\gamma$ of FDPC codes have minimum distance $d$ and the expected value of codewords of weight $w \geq d$ over these codes is at most $\frac{1}{1-\gamma} \cE\{\cA^{(2)}_w\}$. The rest of the proof is the same as the proof for the originial upper bound in \eq{Pe_BEC}. 
\end{proof}

The approximate ML bound of the ensemble of FDPC codes shown in Figure\,\ref{BEC-plot} is with $w_t =20$, i.e., weight distribution is taken into account up to codewords of weight $20$ in \eqref{Pe_BEC3}. Also, they are shortened by $\alpha_4 = 1$ bit and have $d_{\min} \geq 2d = 6$.  

For transmission over B-AWGN with noise variance $\sigma^2$, an upper bound on the performance of the ML decoder, similar to the one in \eq{Pe_BEC}, holds:
\be{Pe_awgn}
P_e^{(\text{ML})} \leq \sum_{w=d_{\min}}^{n} \tilde{A}_w Q(\frac{\sqrt{w}}{\sigma}),
\ee
where, again, a linear code $\tilde{\cC}$ with weight distribution $\{\tilde{A}_w: w=d_{\min},\dots,n\}$ is considered for transmission. Also, $Q(.)$ is the $Q$-function representing the tail of the cumulative probability distribution of normal Gaussian distribution. 

Also, similarly, the result in \Pref{proposition_ml} can be revised to get an upper bound on the probability of ML decoder for shortened FDPC codes for transmission over the B-AWGN channel:
\be{Pe_awgn2}
\cE\{P_e^{(\text{ML})}\} \leq \frac{1}{1-\gamma} \sum_{w=2d}^{n} \cE\{\cA^{(2)}_w\} Q(\frac{\sqrt{w}}{\sigma}),
\ee
where $\gamma$ is defined in \eq{gamma-def}, and $d$ is as in defined in \Pref{proposition_ml}. The approximate ML bound of the ensemble of FDPC codes shown in Figure\,\ref{AWGN-plot} is with $w_t =30$, i.e., weight distribution is taken into account up to codewords of weight $30$ in \eqref{Pe_awgn2}. Also, they are shortened by $\alpha_4 = 1$ bit and have $d_{\min} \geq 2d = 6$.

\section{MP-PL Decoding Algorithm}
\label{sec:dec}

In this section we present the proposed MP-PL decoding algorithm to boost the performance of FDPC codes. This is done separately for BEC channels and the general class of noisy channels, due to the very different nature of how erased/erroneous bits are decoded in an MP-type decoding for these channels. A major building block of the decoding algorithm, in both cases, is a variation of iterative message-passing (MP) decoding algorithm. Another major building block, which is the main novel component, is a \textit{progressive} list decoder that is combined with multiple stages of the MP decoder to form the MP progressive list (MP-PL) decoder, which will be discussed in details in this section. 

\subsection{MP-PL Decoding Algorithm for BEC}
\label{sec:Dec-BEC}

Let $\bc=(c_1,c_2,\dots,c_n)$ denote the transmitted codeword, which is received at the receiver with some erasures. The MP-PL decoder is run in multiple stages, to be discussed next. 

\noindent{\textbf{Standard MP decoder for BEC.}} In each iteration of the MP decoder, any erased variable node which connects to a check node whose other variable node neighbors are already decoded and known, will be turned from erasure to a successfully decoded bit. More specifically, consider an erased bit $c_{i_1}$ and a parity-check equation involving $c_{i_1}$, i.e., 
$$
c_{i_1}+c_{i_2}+\dots+c_{i_{2t}} = 0.
$$
Note that, as highlighted in Section\,\ref{sec:base}, all rows of $H$ have weight equal to $2t$. Then if all other $c_{i_j}$'s, $j=2,\dots,2t$, are already known/decoded, then $c_{i_1}$ will be also decoded. And the same procedure applies to all variable nodes at once in each iteration. 

\noindent{\textbf{A single stage in the MP-PL decoder.}} In each stage, the standard MP decoder is utilized repeatedly for a certain number of iterations $\lambda_{\text{it}}$, or until it saturates, i.e., when no new erasure is decoded from one iteration to the next, whichever comes first ($\lambda_{\text{it}}=4$ is set in our simulations). At this point, the PL decoding component is initiated. In particular, one of the erasures that is not decoded yet is selected, say $c_{j_i}$, in the $i$-th stage, and then the MP decoding continues along the two possible paths, one with $c_{j_i}=0$ and another one with $c_{j_i} =1$ in the stage $i+1$.  We refer to $j_i$ as the $i$-th \textit{path-splitting} index. This doubles the list size as the MP-PL decoder transitions to the next stage. If at any point in the decoding process of a path, an erasure variable node gets two inconsistent decoded results from two neighboring check nodes, that path will be discarded as this indicates that at some point along the progressive list decoding a wrong choice is pursued when splitting the paths. 

\begin{algorithm}[t]
	\caption{MP-PL Decoding Algorithm for BECs} \label{alg_softmap}
	\textbf{Input:} Channel output $\by$ as a codeword $\bc$ with a certain subset of entries erased; the parity-check matrix $H_{r \times n}$; maximum list size $L=2^l$; maximum iterations per stage $\lambda_{\text{it}}$.
	
	\textbf{Output:} transmitted codeword $\bc$

 \textbf{Available functions:} \textit{MP-iter-BEC}: one single iteration of MP decoder over BEC, with inputs $\by$, $H$ and output $\by'$; \\\textit{Path-split-ind}: selecting the path-splitting index, with input $\by$ and $H$, and output index $j$.
 
	\vspace*{0.05in}
	\begin{algorithmic}[1]
        \State initialize: $\tilde{L} = 1$, $\by_1 = \by$.
        \While {$\tilde{L} < L$, $\bc$ not found}

        \For {$i=1,2,\cdots,\tilde{L}$} 

        \For {$j=1,2,\cdots,\lambda_{\text{it}}$} 
        \State $\by_i = $ \textit{MP-iter-BEC}$(\by_i,H)$
        \If {$\by_i$ has no erasure}
        \State $\bc = \by_i$ and break
        \EndIf
        \EndFor
        \State $j$ = \textit{Path-split-ind}$(\by_i,H)$
        \State $\by_{\tilde{L}+i} = \by_i$
        \State $\by_i(j) = 0$
        \State $\by_{\tilde{L}+i}(j) = 1$

        \EndFor
        \State $\tilde{L} = 2\tilde{L}$
	\EndWhile 

	\end{algorithmic}
\end{algorithm}

\noindent{\textbf{Path-splitting index selection criterion.}} The criterion solely depends on the location of current erasures in the current path. Let $H_e$ denote a sub-matrix of the parity-check matrix $H$ by selecting columns of $H$ corresponding to the current erasures. Let $m_1,m_2,\dots,m_r$ denote the number of ones in rows of $H_e$, where $r$ is the number of rows in both $H$ and $H_e$. Let $m_{i'}$ be the minimum non-zero value in the set $\{m_i: i \in [r]\}$. Then $j_i$ is selected as the index of one of the non-zero entries in the $i'$-th row of $H$ that correspond to a current erasure. 

\noindent{\textbf{Remark 3.}} Roughly speaking, the logic behind the selection of the path-splitting indices is to increase the chances of recovering from more erasures when further applying the MP decoding iterations. For instance, if $m_{i'} = 2$, then the way the current path-splitting index is selected ensures that at least another erasure will be corrected in the next iteration of the MP decoder.

\noindent{\textbf{The MP-PL decoder outcome.}} Let $L=2^l$ denote the maximum list size. Then the MP-PL decoder continues up to the $l$-th stage, or till all erasures are decoded in at least one decoding path, whichever comes first. If there is more than one path in which all erasures are corrected, it implies that there is more than one codeword that match with all the non-erased received bits. In other words, the transmitted codeword is not uniquely decodable for such an erasure pattern scenario. 


A brief pseudo-code for the MP-PL decoder, while skipping some detailed steps, is presented as Algorithm\,1. Some of the omitted details include killing decoding paths where some variable nodes received inconsistent messages in the MP decoder, and handling cases with multiple output codewords.

\subsection{MP-PL Decoding Algorithm for General Noisy Channels}
\label{sec:Dec-gen}

The MP-PL decoder for general noisy channels is also run in multiple stages, same as in the BEC case. However, the criterion for selecting path-splitting indices and the how to efficiently check if a codeword is obtained are essentially different. The latter does not apply to the BEC case as all decoding paths are essentially equally likely as long as there is no inconsistency in their MP updates. The MP decoder we utilize in one single stage of the MP-PL decoder is with the min-sum decoder with weighted updates. 

\noindent{\textbf{Weighted min-sum MP decoder.}} In this decoder, all message updates are done while scaled with a small fixed constant $\beta$, e.g., $\beta = 0.05$ is set in our simulations. Furthermore, we also keep updating the LLRs of the coded bits in each round. More specifically, let $\by^{(\text{cur})} = (y_1^{(\text{cur})},\dots,y_n^{(\text{cur})})$ denote the current \textit{soft} information for the coded bits $\bc = (c_1,\dots,c_n)$ in the LLR domain. Let also $q_{i,j}^{(\text{cur})}$ (and $r_{j,i}^{(\text{cur})}$) denote the current message passed from the $i$-th variable node to the $j$-th check node (and vice versa) when $H(j,i) = 1$.  Then the updated soft information $y_i^{\text{(new)}}$, $q_{i,j}^{(\text{new})}$, and $r_{j,i}^{(\text{new})}$ are as follows:
\begin{align}
\label{r-update}
r_{j,i}^{(\text{new})} &= 
\prod_{i' \in \cV_{j \setminus i}} \text{sign}(q_{i',j}^{(\text{cur})}) \min_{i' \in \cV_{j \setminus i}} |q_{i',j}^{(\text{cur})}|,\\
\label{MP-eq}
y_i^{\text{(new)}} &= y_i^{(\text{cur})} + \beta \sum_{j, H(j,i) = 1} r_{j,i}^{(\text{new})},\\
\label{q-update}
q_{i,j}^{(\text{new})} &= y_i^{\text{(new)}} - \beta r_{j,i}^{(\text{new})},
\end{align}
where $\cV_{j \setminus i} := \{i': i' \neq i, H(j,i')=1\}$. In a sense, this version of MP decoding keeps accumulating all the previously generated LLRs in prior iterations while scaling them down in order to not having the accumulated LLRs saturate quickly. 

\noindent{\textbf{Remark 4.}} As it is evident from \eq{MP-eq}, the scaling of $y_i$'s does not affect the decoder. Hence, in cases such as the AWGN or fading channels, the \textit{plain} channel output can be used as the input the the MP-PL decoder in the first stage. This gives an advantage, especially in non-coherent settings, over most soft-decision decoders which require the knowledge of the channel, e.g., the channel gain and noise variance, in order to generate the LLRs to be fed into the decoder.

\noindent{\textbf{Path-splitting index selection criterion.}} The proposed criterion depends solely on which parity-check equations at the end of the current stage are not yet satisfied. More specifically, let $m_1,m_2,\dots,m_r$, with $m_i \in \{0,1\}$, where $m_i =0$ when the $i$-th parity-check equation (corresponding to the $i$-th row if $H$) is satisfied, and $m_i = 1$, otherwise. Then the \textit{deficiency} of the $j$-th variable node is defined as the sum of $m_{j_i}$, over $j_i$ in the set of the indices of its neighboring check nodes. The variable node with the maximum deficiency is selected as the path-splitting index.

In our simulations, we have limited the number of path splittings to $16$, i.e., the maximum list size is $2^{16}$. Potentially, list-pruning strategies can be adopted, same as in the SCL decoding of polar codes \cite{TV}, to ensure the list size remains small. Also, note that this is the worst case scenario, and the average list size is much smaller (e.g., close to $1$ for high-SNRs). Also note that, as shown in Figure\,\ref{AWGN-plot}, the min-sum decoder is already very good. Hence, the list decoding component, which is activated only when the min-sum fails, is only run for a very small fraction of time.  

\noindent{\textbf{Checking paths for codewords and stopping point.}} For each decoding path at the end of each decoding stage, if all the parity-check equations are satisfied, then the current hard decisions are set as the decoder output and the decoder stops. Also, a certain threshold is set on the number of decoding stages, where the MP-PL decoding algorithm stops and a block decoding failure is announced if a codeword is not reached yet. This threshold is set to $16$ in our simulations.

\subsection{Decoding Complexity and Latency}
\label{sec:comp}

Let the complexity of comparison, addition, and multiplication/division in the soft information/LLR domain be denoted by $\mu_c,\mu_a$, and $\mu_d$, respectively, where the unit could be one logical bit operation. Similarly, let the latency of  comparison, addition, and multiplication/addition be denoted by $\nu_c,\nu_a$, and $\nu_d$, respectively, where the unit could be the delay of one logical bit operation. In each iteration, for each parity-check node, the first and the second minimum over its neighboring variable nodes need to be computed. This is done with the complexity of $2\sqrt{n} \mu_c$ and latency of $2 \log n \nu_c$. Each variable node update is done with the complexity of $2s \mu_a$ and latency of $(\log s +2)\nu_a$ for an order-$s$ FDPC code. Then the total complexity of each iteration is $2sn(2\mu_c+\mu_a)$. Assuming full parallelization of all variable and check node operations, which can be done in principle, the total latency of each iteration is $2 \log n \nu_c+(\log s +2)\nu_a$. Since all parameters other than $n$ can be treated as constants, the complexity and latency of each iteration are essentially $O(n)$ and $O(\log n)$, respectively. With list-decoding, the average complexity of the MP-PL decoder scales with the average of $L$, but not the latency. 

To arrive at the complexity/latency comparisons reported in Section\,\ref{sec:Introduction}, note that the SCL decoding of polar codes has complexity and latency of $O(Ln\log n)$ and $O(n)$, respectively. More specifically, the complexity is $nL\log n(\mu_d+\mu_a)$ and the latency is $n(2\nu_d+\nu_a)$. Then assuming a fixed-point implementation with $8$ bits, one can use $\mu_d \approx 8\mu_a$, $\mu_a \approx 2 \mu_c$, where similar relations hold also for $\nu_c,\nu_a$ and $\nu_d$. Of course, the exact numbers would depend on the actual hardware implementation, however, it is expected that the numbers would stay in the same range. With these rough estimates, we arrive at the complexity comparisons reported in Section\,\ref{sec:Introduction}. Note that we use the average number of iterations and list sizes in MP-PL decoder of FDPC codes at SNR$=5.25$\,dB with the performance reported in Figure\,\ref{AWGN-plot} (where it starts to beat polar-CRC with SCL). More specifically, the average list size is $L_{\text{avg}} \approx 1.63$, and the average number of iterations per sample (taking into accounts all paths in the list) is $\approx 5$. 

\section{Further Simulation Results}
\label{sec:sim}

In addition to the simulation results provided in Section\,\ref{sec:Introduction}, in this section, we briefly study a scenario in the \textit{large} block length regime, e.g., with $n=16384$. In this regime, mainly suitable for optical applications, the latency considerations require highly parallezible decoders. As a result, product codes and their other variations, including staircase codes, have been largely studied in this regime in the literature \cite{graell2020forward}. 

We construct an order-$2$ FDPC code at this length with rate $\approx 0.97$. Constructing such a high-rate code can not be accomplished with a reasonable product coding architecture. In fact, even if one considers a product code with $(256,247)$ extended Hamming code as its two components, the overall rate would be $\approx 0.93$. 
In order to just show a baseline, we pick algebraic BCH codes with HDD decoder. Note that polar codes become largely impractical in this regime due to the high-latency of SC-type decoders. The bit-error rate (BER), often considered as the figure of merit in optical settings, of the FDPC code is shown in Figure\,\ref{AWGN-plot2}. We observe that the FDPC code, with the plain min-sum decoder (no list), is $\approx 1$ dB better than the BCH code of the same length and rate. 

Even though it is difficult to obtain simulations results at very low BERs (e.g., BER$\,=10^{-12}$), relevant for optical applications, using Monte Carlo simulations, we can at least comment on the error floor region of FDPC codes. Note that the BLER performance of a code in its error floor region is mainly dominated by the $A_{d_{\min}} Q(\sqrt{d_{\min}}/\sigma)$ term. With $d_{\min}=6$, and $A_{d_{\min}} <20$, as in \eq{eq-avgA6}, the value of $A_{d_{\min}} Q(\sqrt{d_{\min}}/\sigma)$ is approximately $10^{-10}$ at 6\,dB and $2\times 10^{-13}$ at 7\,dB, for FDPC codes. This together with the curve shown in Figure\,\ref{AWGN-plot2} indicate that the error floor of FDPC code happens at very low error rate regime, i.e., when BLER$<10^{-13}$ (or anticipated BER $<10^{-16}$). 

The results and discussions in this section demonstrate that FDPC codes are also very promising candidates for optical applications in high-rate large block length regimes. 

\begin{figure}[t]
	\centering
	\includegraphics[width=\linewidth]{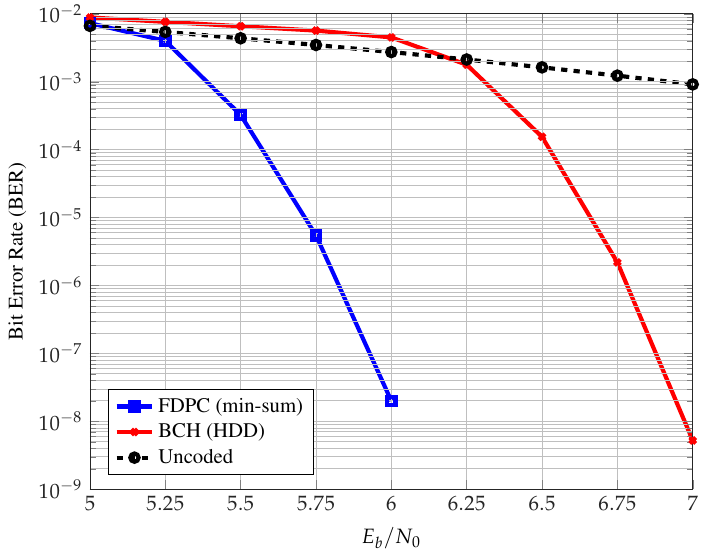}
	\caption{Performance comparison over B-AWGN ($n=16384$, rate$=0.97$)}
	\label{AWGN-plot2}
	\vspace{-0.2in}
\end{figure}

\bibliographystyle{IEEEtran}
\bibliography{ref}

\end{document}